\documentclass[a4paper, twocolumn, reqno]{amsart}
\usepackage[letterpaper, portrait, margin=1in]{geometry}
\usepackage[utf8]{inputenc}
\usepackage{graphicx, caption}
\usepackage{amsfonts}
\usepackage{amsmath,bm}
\usepackage{cite}
\usepackage{amssymb}
\usepackage{amsthm}
\usepackage{color}
\usepackage{mathtools}
\usepackage{braket}
\usepackage{enumitem}
\usepackage{multirow}
\usepackage{hyperref}
\usepackage{qcircuit}
\usepackage{float}

\usepackage{subfig}
\usepackage{amsaddr}
\usepackage{qcircuit}
\usepackage{tikz}
\usepackage{yquant}
\usepackage{balance}

\theoremstyle{plain}
\newtheorem{thm}{Theorem}%[section] 

\newtheorem{defn}{Definition} % definition numbers are dependent on theorem numbers
\newtheorem{proposition}[thm]{Proposition}

\usepackage{mathtools}

\title{Every Classical Sampling Circuit is a Quantum Sampling Circuit} 
\author{Steven Herbert$^*$}
\address{Cambridge  Quantum  Computing  Ltd,  9a  Bridge  Street,  Cambridge,  CB2  1UB,  UK \\ Department of Computer Science and Technology, University of Cambridge, UK}
\begin{document}

\onecolumn

\begin{abstract}
\noindent 
This note introduces ``Q-marginals'', which are quantum states encoding some probability distribution in a manner suitable for use in Quantum Monte Carlo Integration (QMCI), and shows that these can be prepared directly from a classical circuit sampling for the probability distribution of interest. This result is important as the quantum advantage in Monte Carlo integration is in the form of a reduction in the number of uses of a quantum state encoding the probability distribution (in QMCI) relative to the number of samples that would be required in classical MCI -- hence it only translates into a \textit{computational} advantage if the number of operations required to prepare this quantum state encoding the probability distribution is comparable to the number of operations required to generate a classical sample (as the Q-marginal construction achieves).
\end{abstract}

\twocolumn[{%
 \centering
 \maketitle
 
 \vspace{-1cm}
}]

\section{Introduction}

Loading classical data on to a quantum computer is one of the major remaining barriers to using quantum computers to give computational advantage in data science tasks. These data could be in the form of ``raw'' data, to be stored in a quantum memory (\textit{e.g.}, ``QRAM''  \cite{qRAM}), or in the form of some probability distribution to be sampled from. In this note we focus on the latter, specifically, the preparation of states that sample from some probability distribution of interest when measured in the computational basis. That is, states of the form:
\begin{equation}
\ket{p} = \sum_{i \in \{0,1\}^n } \phi_i \sqrt{p_i} \ket{i}
\end{equation}
where $\{ p_i \}$ is some discrete $2^n$-point probability distribution, whose support is (implicitly) mapped to the binary strings $0^n, \dots, 1^n$, and $\phi_i$ denotes a relative phase, \textit{i.e.}, $\phi_i$ is a complex number such that $|\phi_i| = 1$, and so the same distribution is sampled from regardless of the values of these relative phases. We may further define $P$ to be a circuit that prepares $\ket{p}$ from $\ket{0^n}$, that is $\ket{p}= P \ket{0^{n}}$. States of the form $\ket{p}$ are sometimes referred to as ``Q-samples'' as, even though by specifying the relative phases we can see that any quantum state can be expressed in this manner, the construction such that computational basis measurements \textit{sample} from the probability distribution $\{ p_i \}$ is the pertinent feature. Circuits preparing Q-samples are important for a variety of application in quantum data-science, and in particular we are interested in preparing Q-samples for Quantum Monte Carlo Integration \cite{MontanaroMC, herbert2021quantum}\makeatletter{\renewcommand*{\@makefnmark}{}
\footnotetext{Contact: Steven.Herbert@cambridgequantum.com}\makeatother}

The quantum advantage in Monte Carlo integration is essentially a query (or sample) complexity advantage: to converge to some specified mean squared error, QMCI requires (asymptotically) quadratically fewer uses of the circuit $P$ than the number of samples from $\{ p_i \}$ that are required classically. However, this leaves open the question of how to construct such a circuit, $P$. For many years, it was supposed that the Grover Rudolph method \cite{grover2002creating} could be used to prepare Q-samples for a number of commonly-used probability distributions, but recently it was shown that this approach is not sufficient to uphold the quantum advantage when a full audit of the required operations is undertaken \cite{HerbertGR}. 

In this note we seek to put classical and quantum Monte Carlo integration on an equal footing by first considering the computational cost of generating a classical sample. In particular we use the fact that even in classical Monte Carlo (executed on a digital computer) a sample from $\{ p_i \}$ still requires one use of a circuit, namely the ``classical sampling circuit''.

\begin{defn}
A classical sampling circuit for some probability distribution, $\{ p_i \} $, is a map $f: \{0,1\}^m \to \{0,1\}^n$, which is such that $\frac{1}{2^m} |\{x : f(x) = X \}| = p_X$. That is, when the input is a uniformly random bitstring, then we sample from $\{ p_i \}$.
\end{defn}
This is a completely general way to describe the (digital) computational process of classical sampling -- even if the output sample is interpreted as a sample from a multivariate distribution. Furthermore, even when sampling from a continuous and / or infinitely-supported probability distribution, if a digital computer is used then suitable quantisation and / or truncation will always be required, and this too is captured in the notion of a classical sampling circuit. This definition is also such that, for example, a neural network trained to sample from some desired distribution given a uniformly random input, is an instance of a ``classical sampling circuit''. In order to connect classical sampling circuits to quantum sampling, we need also to introduce ``Q-marginals'', which are closely related to Q-samples:
\begin{defn}
A ``Q-marginal'' for the probability distribution $\{ p_i \}$ is a quantum state of the form:
\begin{equation}
\ket{\tilde{p}} = \sum_{i \in \{0,1\}^n} \, \sum_{\{ j \in \{0,1 \}^m : f(j) = i \} } a_{j} \ket{j} \! \ket{i} 
\end{equation}
where for all $i$ we have $\sum_{\{ j \in \{0,1 \}^m : f(j) = i \} } | a_{j}|^2 = p_i$ (\textit{i.e.} for any suitable function, $f$).
\end{defn}
That is, we group together all of the terms in the first register that are entangled with the \textit{same} value of $i$, and we can see that measuring the second register does indeed sample from $p_i$. By deploying two registers in this way, the probability distribution of interest is now a \textit{marginal} distribution of some larger joint distribution -- hence the name. The application of QMCI naturally handles the situation where the Monte Carlo integral concerns the marginal distribution of some joint multivariate distribution \cite[Section~II]{herbert2021quantum}, and hence we can see that preparing Q-marginals rather than Q-samples is sufficient for QMCI.

The main result of this note is a proof that a circuit preparing Q-marginals can always be constructed from the corresponding classical sampling circuit.

\section{Q-Marginals from Classical Samples}

\begin{proposition}
Let $U$ be a reversible form of some classical sampling circuit (in the sense of Definition 1) for some probability distribution, $\{p_i \}$, then we can construct a circuit, $U'$ to prepare a Q-marginal of $\{ p_i \}$ using only $U$ and $m$ Hadamard gates in a single layer.
\end{proposition}
\begin{proof}
We construct $U'$ such that:
\begin{equation}
U' =   U (H^{\otimes m} \otimes I_2^{\otimes n})
\end{equation}
where $H$ is the Hadamard gate, and $I_2$ is the $2 \times 2$ identity, as shown in Fig.~\ref{f1}. Using the fact that the reversible form of the sampling circuit is such that $\ket{x} \! \ket{y} \xrightarrow[]{U} \ket{x}  \! \ket{ y \oplus f(x)}$ (for any computational basis state, $\ket{x}$), we thus have that:
\begin{align}
\ket{0^{m+n}} \xrightarrow[]{U'} & \frac{1}{\sqrt{2^m}}  \sum_{j \in \{ 0,1 \}^m }\ket{j} \! \ket{f(j)} \nonumber \\
= &  \sum_{i \in \{ 0,1 \}^n }    \sum_{ \{ j: f(j) = i \} }  \frac{1}{\sqrt{2^m}} \ket{j} \! \ket{i}
\end{align}
which from Definition~2 is a Q-marginal with $a_j = \frac{1}{\sqrt{2^m}} $ for all $j$, and from Definition~1 we have that $ \frac{1}{2^m} | \{ j: f(j) = i \} | = p_i$, and so this is indeed a Q-marginal of $\{ p_i \}$.
\end{proof}

 \begin{figure}[t!] %H if fixed

\centering

\begin{tabular}{ c }
%\multicolumn{3}{c}{
\begin{tikzpicture}
    \begin{yquant}
    qubit {$\ket{x}$} r[1];
    qubit {$\ket{y}$} q[1]; %$\sum\limits_{i \in \{0,1 \}^m} \sqrt{\mathrm{p}^{(m)}_i} \ket{i}$
    slash q[0];
    slash r[0];
    box {$U$} (q[0],r[0]) ;
    slash q[0];
    slash r[0];
    output {$\ket{x}$} r[0];
    output {$\ket{y \oplus f(x) }$} q[0]; %$\sum\limits_{i \in \{0,1 \}^{m+1}} \sqrt{\mathrm{p}^{(m+1)}_i} \ket{i}$
  \end{yquant}
  \end{tikzpicture}
  \\
  (a) \\
  { } \\
  \begin{tikzpicture}
    \begin{yquant}
    qubit {$\ket{0^m}$} r;
    qubit {$\ket{0^n}$} r[+1]; %$\sum\limits_{i \in \{0,1 \}^m} \sqrt{\mathrm{p}^{(m)}_i} \ket{i}$
    slash r[0];
    slash r[1];
    box {$H^{\otimes m}$} r[0] ;
    box {$U$} (r[0],r[1]) ;
    slash r[0];
    slash r[1];
    output {$\frac{1}{\sqrt{2^m}}\sum_{j \in \{0,1\}^m} \ket{j} \ket{f(j)}$}  (r);
    %output {$\sum_i \ket{i}$} r[0];
    %output {$\sum_i \ket{i} \ket{f(i)}$} q[0]; %$\sum\limits_{i \in \{0,1 \}^{m+1}} \sqrt{\mathrm{p}^{(m+1)}_i} \ket{i}$
  \end{yquant}
  \end{tikzpicture}\\
   (b) \\
  \end{tabular}
\captionsetup{width=.9\linewidth}
	\caption{(a) Reversible sampling circuit; (b) the circuit, $U'$, to prepare Q-marginals, shown here acting on the input state $\ket{0^{m+n}}$.}
\label{f1}
\end{figure}

\section{Discussion}

The result is an extremely simple one, and is essentially a corollary of the basic principle that it is always possible to construct a reversible version of any classical circuit. Indeed, Proposition~1 says that Q-marginals can be prepared by applying a reversible circuit to a circuit preparing Q-samples, and in particular if the Q-sampling circuit prepares the uniform distribution (\textit{i.e.} it is a single bank of Hadamard gates) then a reversible version of a classical sampling circuit can be used to prepare Q-marginals of the same distribution. Nevertheless, there is merit in explicitly connecting this to the fact that Q-marginals (as opposed to Q-samples) suffice for QMCI. For this construction shows that a suitable quantum state encoding \emph{any} probability distribution can be constructed directly from the classical circuit required to sample from that distribution, and hence that there is \emph{always} an actual \emph{computational} advantage when performing QMCI. We no longer need to appeal to vague arguments about the complexity of preparing Q-samples in order to make this claim. Notably, we also inherit any quantisation and truncation that would be made to approximately sample from a continuous probability distribution on a classical computer, and hence there is no risk of hidden complexity in the process of preparing quantum samples to the same accuracy as the corresponding classical samples (as is the central flaw in using the Grover-Rudolph algorithm for Q-sample preparation \cite{HerbertGR}).

An important question to ask next, is when is such a construction useful in practice? In cases where we already have effective classical sampling algorithms, then using the reversible circuit thereof (as a starting point at least) may well prove to be a fruitful approach. Another question is how many of the operations performed in $U$ can be pushed into classical post-processing, using a Fourier series decomposition of the Monte Carlo integral \cite{herbert2021quantum}.

%\pagebreak

\section*{Acknowledgement}

Thanks to Ross Duncan and Cristina Cirstoiu for reviewing this note.

\balance

%\bibliography{mybib}{}
%\bibliographystyle{IEEEtran}

\end{document}